\documentclass[11pt]{article}

\usepackage{color}
\usepackage{amssymb}
\usepackage{amsthm}
\usepackage{amsmath}
\usepackage{fullpage,algorithmic}
\usepackage{graphicx,url,amsmath,amsthm,amsfonts,amssymb,subfigure,bbm}
\usepackage{algorithm, thmtools,thm-restate,enumerate,verbatim,bm}
\usepackage{pdfsync}
\usepackage[usenames,dvipsnames,svgnames,table]{xcolor}
\definecolor{darkgreen}{rgb}{0.0,0,0.9}
\usepackage[pagebackref,letterpaper=true,colorlinks=true,pdfpagemode=none,citecolor=OliveGreen,linkcolor=BrickRed,urlcolor=BrickRed,pdfstartview=FitH]{hyperref}

\DeclareMathOperator{\Exp}{\mathbb{E}}

\renewcommand{\P}[1]{{\mathbb{P}}\left[#1\right]}
\newcommand{\hP}[1]{{\bf P}\left[#1\right]}
\newcommand{\PP}[2]{{\mathbb{P}}_{#1}\left[#2\right]}
\newcommand{\hPP}[2]{{\bf P}_{#1}\left[#2\right]}
\newcommand{\cPP}[2]{{\cal{P}}_{#1}\left[#2\right]}
\newcommand{\kPP}[3]{{\bf P}^{#2}_{#1}\left[#3\right]}

\newcommand{\E}[1]{{\Exp}\left[#1\right]}
\newcommand{\hE}[1]{{\bf E}\left[#1\right]}
\newcommand{\EE}[2]{\Exp_{#1}\left[#2\right]}
\newcommand{\hEE}[2]{{\bf E}_{#1}\left[#2\right]}
\newcommand{\kEE}[3]{{\bf E}^{#2}_{#1}\left[#3\right]}
\newcommand{\tp}[1]{{#1'}}

\providecommand{\norm}[1]{\lVert#1\rVert}

\declaretheorem[numberwithin=section]{theorem}
\declaretheorem[sibling=theorem]{lemma}
\declaretheorem[sibling=theorem]{proposition}
\declaretheorem[sibling=theorem]{claim}

\def\bx{{\bf x}}
\def\by{{\bf y}}
\def\bz{{\bf z}}
\def\bv{{\bf v}}
\def\bp{{\bf p}}

\def\bpsi{{\bm{\pi}}}
\def\b1{{\bf 1}}
\def\eps{{\epsilon}}
\def\cD{{\cal D}}

\def\G{{\cal G}}
\DeclareMathOperator{\rem}{rem}
\DeclareMathOperator{\esc}{esc}
\DeclareMathOperator{\rank}{rank}
\DeclareMathOperator{\tr}{Tr}

\DeclareMathOperator{\polylog}{polylog}
\def\gensam{{\textup{\texttt{GenerateSample}}}}
\def\Evo{{\textup{\texttt{ParESP}}}}
\def\SSE{{\textup{\texttt{Threshold}}}}
\def\nibble{{\textup{\texttt{Nibble}}}}
\begin{document}

\title{Approximating the Expansion Profile and \\
Almost Optimal Local Graph Clustering}
\author{Shayan Oveis Gharan
\thanks{Department of Management Science and Engineering, Stanford University. Supported by a Stanford Graduate Fellowship. Email:\protect\url{shayan@stanford.edu}.}
\and
Luca Trevisan\thanks{Department of Computer Science, Stanford University. This material is based upon  work supported by the National Science Foundation under grant No.  CCF 1017403.
Email:\protect\url{trevisan@stanford.edu}.}
}

\date{}
\maketitle
\begin{abstract}
Spectral partitioning is a simple, nearly-linear time, algorithm to
find sparse cuts, and the Cheeger inequalities provide a worst-case
guarantee for the quality of the approximation found by the
algorithm. Local graph partitioning algorithms
\cite{ST08,ACL06,AP09} run in time that is nearly linear in the size
of the output set, and their approximation guarantee is worse than
the guarantee provided by the Cheeger inequalities by a
poly-logarithmic $\log^{\Omega(1)} n$ factor. It has been an open
problem to design a local graph clustering algorithm with an
approximation guarantee close to the guarantee of the Cheeger
inequalities and with a running time nearly linear in the size of
the output.

 In this paper we solve this problem; we design an algorithm with the same guarantee (up to a constant factor) as the Cheeger inequality, that runs in time slightly super linear in the size of the output.
  This is the first sublinear (in the size of the input) time algorithm with almost the same guarantee
 as the Cheeger's inequality.
  As a byproduct of our results, we prove a bicriteria approximation algorithm for
 the expansion profile of any graph.
 Let $\phi(\gamma) = \min_{\mu(S) \leq \gamma}\phi(S)$. There is a polynomial time algorithm that, for any $\gamma,\eps>0$,  finds a set $S$ of volume $\mu(S)\leq 2\gamma^{1+\eps}$, and expansion $\phi(S)\leq \sqrt{2\phi(\gamma)/\eps}$.
Our proof techniques also provide a simpler proof of the structural result of Arora, Barak, Steurer \cite{ABS10}, that can be applied to irregular graphs.

Our main technical tool is that for any set $S$ of vertices of a graph, a lazy $t$-step random walk started from a randomly chosen vertex of $S$, will remain entirely inside $S$ with probability at least $(1-\phi(S)/2)^t$.
This itself provides a new lower bound to the uniform mixing time of any finite states reversible markov chain.
\end{abstract}

\section{Introduction}
Let $G=(V,E)$ be an undirected graph, with $n:=|V|$ vertices, and let $d(v)$ denote the degree of vertex $v\in V$. The {\em measure (volume)} of a set $S\subseteq V$ is defined as the sum of the degree of vertices in $S$,
$$\mu(S):=\sum_{v\in S} d(v).$$
  The {\em conductance} of a set $S$ is defined as
$$ \phi(S) := \partial(S)/\mu(S)$$
where $\partial(S)$ denotes the number of edges that leaves $S$.
Let
$$\phi (G) :=\min_{S : \mu(S) \leq \mu(V)/2}  \phi(S) $$

be the conductance (uniform sparsest cut) of $G$. The Cheeger
inequalities \cite{AM85,Alon86} prove that the spectral partitioning
algorithms finds, in nearly linear time, a  $O(1/\sqrt{\phi(G)})$
approximation to the uniform sparsest cut problem. Most notably, the
approximation factor does not depend of the size of the graph; in
particular the Cheeger inequalities imply a constant factor
approximation, if $\phi(G)$ is  constant. Variants of the spectral
partitioning algorithm are widely used in practice
\cite{Kleinberg99,SM00,TM06}.

Often, one is interested in applying a sparsest cut approximation algorithm {\em iteratively}, that is, first
find an approximate sparsest cut in the graph, and then recurse on one or both of the subgraphs
induced by the set found by the algorithm and by its complement. Such iteration might be used
to find a {\em balanced} sparse cut if one exists (c.f. \cite{OSV12}), or to find a good {\em clustering} of the graph, an
approach that lead to approximate clusterings with good worst-case guarantees, as shown by
Kannan, Vempala and Vetta \cite{KVV04}. Even though each application of the spectral partitioning
algorithm runs in nearly linear time, iterated applications of the algorithm can result in a quadratic
running time. 

Spielman and Teng \cite{ST04}, and subsequently \cite{ACL06, AP09}  studied local graph partitioning algorithms that find a set $S$ of approximately minimal conductance in time nearly linear in the size of the output set $S$. Note that the running time can be sub linear in the size of the input graph if the algorithm finds a small output set $S$.
When iterated, such an algorithm finds a balanced sparse cut in nearly linear time in
the size of the graph, and can be used to find a good clustering in nearly linear time as well.

Another advantage of such ``local'' algorithms is that if there are both large and small sets of
near-optimal conductance, the algorithm is more likely to find the smaller sets. Thus, such algorithms
can be used to approximate the ``small-set expander'' problem, which is related
to the unique games conjecture \cite{RS10} and the expansion profile of a graph (that is,
what is the cut of smallest conductance among all sets of a given volume). Finding small, low-conductance,
sets is also interesting in clustering applications. In a social network, for example, a low-conductance
set of users in the ``friendship'' graph represents a ``community'' of users who are significantly more likely
to be friends with other members of the community than with non-members, and discovering
such communities has several applications. While large communities might correspond to large-scale,
known, factors, such as the fact that American users are more likely to have other Americans as friends,
or that people are more likely to have friends around their age, small communities contain more
interesting information.

A {\em local graph clustering} algorithm, is a {\em local} graph algorithm
that finds a non-expanding set in the local neighborhood of a given vertex $v$, in time proportional to the size of the
output set. The {\em work/volume ratio} of such an algorithm, which is the ratio of the the computational time of the algorithm in a single run, and the volume of the output set, may depend
only poly logarithmically to the size of the graph.

The problem first studied
in the remarkable work of Spielman and Teng \cite{ST04}. Spielman
and Teng design an algorithm \nibble~such that for any set $U\subseteq V$,
if the initial vertex, $v$, is sampled randomly according to the degree of vertices in $U$, with a constant probability, \nibble~finds a set of conductance $O(\phi^{1/2}(U)\log^{3/2}{n})$, with a work/volume ratio of  $O(\phi^{-2}(U)\polylog(n))$,
\nibble~finds the desired set by looking at the {\em threshold sets} of the probability distribution of a $t$-step random walk started at $v$. To achieve the desired computational time they keep the support of the probability distribution small by removing a small portion of the probability mass at each step.

Andersen, Chung and Lang \cite{ACL06}, used the approximate PageRank vector rather than approximate random walk distribution, and they
managed to improve the conductance of the output set to $O(\sqrt{\phi(U)\log{n}})$, and the work/volume ratio to $O(\phi^{-1}(U)\polylog{n})$. More recently, Andersen and Peres \cite{AP09}, use the {\em evolving set process} developed in the work of Diaconis and Fill \cite{DF90}, and they improved the work/volume ratio to $O(\phi^{-1/2}(U) \polylog{n})$,
while achieving the same guarantee as \cite{ACL06} on the conductance of the output set.

It has been a long standing open problem to design a local variant of the Cheeger's inequalities: that is to provide a {\em sublinear} time algorithm with an approximation guarantee that does not depend on the size of $G$, assuming that the size of the optimum set is sufficiently smaller than $n$, and a randomly chosen vertex of the optimum set is given.
In this work we answer this question, and we prove the following theorem:
\begin{restatable}{theorem}{ssefast}
\label{thm:ssefast} $\Evo(v,\gamma,\phi,\eps)$ takes as input a
starting vetex $v\in V$, a target conductance $\phi\in (0,1)$, a
target size $\gamma$, and $0<\eps<1$. For a given run of the
algorithm it outputs a set $S$ of vertices with the expected work
per volume ratio of $O(\gamma^{\eps} \phi^{-1/2}\log^2{n})$. If
$U\subseteq V$ is a set of vertices that satisfy $\phi(U)\leq \phi$,
and $\mu(U)\leq \gamma$, then there is a subset $U'\subseteq U$ with
volume at least $\mu(U)/2$, such that if $v\in U'$, with a constant
probability $S$ satisfies,
\begin{enumerate}
\item $\phi(S) = O(\sqrt{\phi/\eps})$,
\item $\mu(S) \leq O(\gamma^{1+\eps})$.
\end{enumerate}
\end{restatable}
We remark that unlike the previous local graph clustering algorithms, the running time of the algorithm is slightly super linear in the size of the optimum.

As a byproduct of the above result we give an approximation algorithm for  the  expansion profile of $G$.   Lovasz and Kannan \cite{LK99} defined the {\em expansion profile} of a graph $G$ as follows:
$$ \phi(\gamma) := \min_{S:\mu(S)\leq \gamma} \phi(S).$$
Lovasz and Kannan used expansion profile as a parameter to prove strong upper-bounds on the mixing time of random walks. The notion of expansion profile recently received significant attention in the literature
because of its close connection to the small set expansion problem and the unique games conjecture \cite{RS10}.
Raghavendra, Steurer, Tetali \cite{RST10}, and Bansal et al.~\cite{BFKM11}
use semidefinite programming and designed  algorithms that approximate $\phi(\gamma)$ within $O(\sqrt{\phi(\gamma)^{-1}\log\frac{\mu(V)}{\gamma}})$,
and $O(\sqrt{\log{n}\log\frac{\mu(V)}{\gamma}})$ of the optimum, respectively. However,
in the interesting regime of $\gamma=o(\mu(V))$, which is of interests to the small set expansion problem, the quality of both  approximation algorithms is  not independent of $\gamma$.


Here, we prove $\gamma$ independent approximation of $\phi(\gamma)$ as a function  of $\phi(\gamma^{1-\eps})$, without any dependency in the size of the graph; specifically we prove the following theorem:
\begin{restatable}{theorem}{mainsse}
\label{thm:sse}
There is a polynomial time algorithm that takes as input a target conductance $\phi$, and $0<\eps<1/4$, and outputs a set $S$, s.t. if $\phi(U)\leq \phi$, for $U\subseteq V$, then
 $\mu(S)\leq 2\mu(U)^{1+\eps}$, and  $\phi(S) \leq \sqrt{2\phi/\eps}$.
\end{restatable}
Our theorem indicates that the hard instance of the small set expansion problem are those where $\phi(\gamma)\approx 1$, for  $\gamma\leq n^{1-\Omega(1)}$.

We remark that one can also use \Evo~to  approximate $\phi(\gamma)$  with slightly worse guarantees (up to constant factors), in sub linear time. Here, for the sake of clarity and simplicity of the arguments, we prove the theorem using random walks.
 Independent of our work, Kwok and Lau \cite{KL12} have obtained  a somewhat different proof of \autoref{thm:sse}.

Our analysis techniques also provide a simpler proof of the structural result of Arora, Barak, Steurer \cite{ABS10}, that can be applied to non-regular graphs.
Let $A$ be the adjacency matrix of $G$, and $D$ be the diagonal matrix of vertex degrees. Let the {\em threshold rank} of $G$, denoted by $\rank_{1-\eta}(D^{-1}A)$, be the number (with multiplicities) of eigenvalues $\lambda$ of $D^{-1}A$, satisfying $\lambda > 1-\eta$.

\begin{restatable}{theorem}{mainabs}
\label{thm:abs}
For any graph $G$, and $0<\eps\leq1$, if $\rank_{1-\eta}(D^{-1}A) \geq n^{(1+\eps)\eta/\phi}$, then there exists  a set $S\subseteq V$ of volume $\mu(S) \leq 4\mu(V)n^{-\eta/\phi}$
and $\phi(S)\leq \sqrt{2\phi/\eps}$. Such a set can be found by finding the smallest threshold set  of conductance $\sqrt{2\phi/\eps}$, among the rows of $(D^{-1}A)^t$, for $t=O(\log{n}/\phi)$.
 \end{restatable}
We remark that Arora et al.~\cite{ABS10} prove a variant of the above theorem for regular graphs with the stronger assumption that $\rank_{1-\eta}(D^{-1}A) \geq n^{100 \eta/\phi}$. This essentially resolves their question of whether the factor $100$ can be improved to $1+\eps$. Independent of our work, O'Donnell and Witmer \cite{OW12} obtained a different proof of the above theorem.

 \subsection{Techniques}

Our main technical result is that if $S$ is a set of vertices, and we consider a $t$-step {\em lazy} random walk
started at a random element of $S$, then the probability that the walk is entirely contained in $S$
is at least $(1-\phi(S)/2)^t$. Previously, only the lower bound $1-t\phi(S)/2$ was known, and the
analysis of other local clustering algorithms implicitly or explicitly depended on such a bound.

For comparison, when $t=1/\phi(S)$, the known bound would imply that the walk has probability
at least $1/2$ of being entirely contained in $S$, with no guarantee being available in the case
$t=2/\phi(S)$, while our bound implies that for $t= (\alpha \ln n)/\phi$ the probability  of being
entirely contained in $S$ is still at least $1/n^\alpha$. Roughly speaking, the $\Omega(\log n)$ factor that we gain in the length of walks that we can study corresponds to our improvement in the expansion bound, while the $1/n^\alpha$ factor that we lose in the probability corresponds to the factor that
we lose in the size of the non-expanding set.
We also use this bound to prove stronger lower bounds on the uniform mixing time of reversible markov chains.

Our polynomial time algorithm to approximate the expansion profile of a graph is the same
as the algorithm used by Arora, Barak and Steurer \cite{ABS10} to find small non-expanding
sets in graphs of a given threshold rank, but our analysis is different. (Our analysis can be also
be used to give a different proof of their result.) Arora, Barak and Steurer use the threshold rank to argue that a random walk started from a random vertex of $G$ will be at the initial vertex after $t$ steps with probability at least $\rank_{1-\eta}(G) (1-\eta)^t/n$. Then, they argue that if   all sets of a certain size have large conductance, this probability must be small, which is a contradiction. To make this quantitative they use the second norm of the probability distribution vector ($\norm{\bp_t}$)  as a potential function, where $\bp_t$ is the distribution of the walk after $t$ steps, and they choose $t$  to get a sufficiently small potential function and argue that the probability of being at the initial vertex after $t$ steps must be small. In our analysis, we use the fact that for any set $S$, there is a vertex $v\in S$ such that the probability that the walk started at $v$  remains in $S$ is  at least
$(1-\phi(S)/2)^t$. Then, we use the potential function $I(\bp_t,\gamma)$ introduced in the work of Lovasz and Simonovits \cite{LS90}. Roughly speaking, $I(\bp_t,\gamma)$ is defined as follows: consider the distribution $\bp_t$ of the vertex reached in a $t$-step random walk started from a random element of
 $S$, take the  $k$ vertices of highest probability under $\bp_t$, where $k$ is chosen so that their
 total volume is about $\gamma$, then $I(\bp_t,\gamma)$ is the total probability under $\bp_t$
 of those $k$ vertices.

Using the machinery of Lovasz and Simonovits, we can upper-bound  $I(\bp_t,\gamma)$  by $\frac{\gamma}{\Gamma} + \sqrt{\gamma}(1-\phi^2/2)^t$ conditioned on all of the threshold sets of volume at most $\Gamma$ of the probability distribution vectors up to time $t$ having conductance less than $\phi$.
Letting $t=\Omega(\alpha\log\mu(S)/\phi(S))$,  $\Gamma = O(\mu(S)^{1+\alpha})$, and $\phi=O(\sqrt{\phi(S)/\alpha})$ since the walk remains in $S$ with probability $\mu(S)^{-\alpha} < I(\bp_t,\gamma)$, at least one of the threshold sets of volume at most $O(\mu(S)^{1+\alpha})$, must have  conductance $O(\sqrt{\phi(S)/\alpha})$.

Our local algorithm uses the evolving set process. The evolving set process starts with a vertex $v$
of the graph, and then produces a sequence of sets $S_1,S_2,\ldots, S_\tau$, with the property that
at least one set $S_t$ is such that $\partial S_t/\mu(S_t) \leq O(\sqrt {\log \mu(S_\tau) /\tau})$. 
If one can show
that up to some time $T$ the process constructs sets all of volume at most $\gamma$, then
we get a set of volume at most $\gamma$ and conductance at most $O(\sqrt {\log\gamma/T})$. Andersen
and Peres were able to show that if the graph has a set  $S$ of conductance $\phi$, then the process is likely
to construct sets all of volume at most $2 \mu(S)$ for at least
$T=\Omega(1/\phi)$ steps, if started from a random element of the $S$, leading to
their $O(\sqrt{ \phi \log n})$ guarantee. We show that for any chosen $\alpha<1/2$, the process will construct sets of volume at
most $O(\mu(S)^{1+\alpha})$ for $T=\Omega( \alpha \log \mu(S) /\phi)$ steps, with probability
at least $1/\mu(S)^{\alpha}$. This is enough to guarantee that, at least with probability
$1/\mu(S)^\alpha$, the process constructs at least one set of conductance $O(\sqrt { \phi/\alpha } )$.
To obtain this conclusion, we also need to strengthen the first part of the analysis of
Andersen and Peres: they show that the process has at least a constant probability
of constructing a set of low conductance in the first $t$ steps, while we need to show
that this happens with probability at least $1-1/\mu(S)^{\Omega(1)}$, because we need to take
a union bound with the event that $t$ is large, for which probability we only have a $\mu(S)^{-\Omega(1)}$
lower bound.
Finally, to achieve a constant probability of success, we run $\mu(S)^\alpha$ copies of the evolving set process simultaneously, and stop as soon as one of the copies finds a small non-expanding set.

\section{Preliminaries}

\subsection{Notations}
Let $G=(V,E)$ be an undirected graph, with $n:=|V|$ vertices and $m:=|E|$ edges.
Let $A$ be the adjacency matrix of $G$, $D$ be the diagonal matrix of vertex degrees, and $d(v)$ be the degree of vertex $v\in V$.
The {\em volume} of a subset $S\subseteq V$ is defined as the summation
of the degree of vertices in $S$,
$$\mu(S):=\sum_{v\in S} d(v).$$
Let $E(S,V\setminus S):=\left\{\{u,v\}: u\in S, v\notin S\right\} $ be the set of the edges connecting $S$ to $V\setminus S$, and we use $\partial(S)$ to denote
the number of those edges, we also let $E(S):=\left\{\{u,v\}: u,v\in S\right\}$ be the set of edges
inside $S$.  The {\em conductance} of a set $S\subseteq V$ is defined to be
$$ \phi(S) := \partial(S)/\mu(S).$$
Observe that $\phi(V)=0$. In the literature, the conductance of a set is sometimes defined to be
$\partial(S)/\min(\mu(S),\mu(V\setminus S))$. Notice that the quantities
are within a constant factor of each other if $\mu(S)=O(\mu(V\setminus S))$. Since, here we
are interested in finding {\em small} non-expanding sets, we would rather work with the above definition.

We define the following  probability distribution vector on a set $S\subseteq V$ of vertices:
\begin{eqnarray*}
\pi_S(v) := \begin{cases}
d(v)/\mu(S) & \textrm{if $v\in S$,}\\
0 & \textrm{otherwise.}
\end{cases}
\end{eqnarray*}
In particular, we use $\pi(v) \equiv \pi_V(v)$ as the stationary distribution of a random  walk in $G$.

Throughout the paper, let $I$ be the identity matrix, and
for any subset $S\subseteq V$, let $I_S$ be the diagonal matrix such that $I_S(v,v)=1$, if $v\in S$ and $0$ otherwise.
Also, let $\b1$ be all one vector, and $\b1_S$ be the indicator vector of the set $S$.  We may abuse the notation and use $\b1_v$
instead of $\b1_{\{v\}}$ for a vertex $v\in V$.

We use lower bold letters to denote the vectors, and capital letters for the matrices/sets.
For a vector $\bx: V\rightarrow R$, and a set $S\subseteq V$, we use $x(S):=\sum_{v\in S} x(v)$. Unless otherwise specified, $\bx$ is considered to be a column vector, and $\bx'$ is its transpose. 

For a square matrix $A$, we use $\lambda_{\min}(A)$ to denote the minimum eigenvalue of $A$, and $\lambda_{\max}(A)$ to denote  the maximum eigenvalue of $A$.

\subsection{Random Walks}
We will consider the {\em lazy random walk} on $G$ that each time step stays at the current vertex with probability $1/2$ and otherwise moves to the endpoint of a random edge attached
to the current vertex. We abuse notation, and we use $\G:=(D^{-1}A + I )/2$ as the transition probability matrix of this random walk, and $\pi(.)$ is the unique stationary distribution, that is $\tp{\bpsi}\G=\tp{\bpsi}$.
We write $\cPP{v}{.}$ to denote the probability measure of the lazy random walk started from a vertex $v\in V$.

Let $X_t$ be the random variable indicating the $t^{th}$ step of the random walk started at $v$. Observe that
the distribution of $X_t$ is exactly, $\tp{\b1_{v}} \G^t$.
For a subset $S\subseteq V$, and $v\in V$, and integer $t>0$, we write $\esc(v, t, S):=\cPP{v}{\cup_{i=0}^t X_i\notin S}$ to denote  the probability that the random walk
started at $v$ leaves $S$ in the first $t$ steps, and $\rem(v,t,S):=1-\esc(v,t,S)$ as the probability that the walk stays entirely inside $S$. It follows that,
\begin{equation}
\label{eq:remprob} \rem(v,t,S) := \tp{\b1_{v}} (I_S \G I_S)^t \b1_S.
\end{equation}

\subsection{Spectral Properties of the Transition Probability Matrix}
Although the transition probability matrix, $\G$, is not a symmetric
matrix, it features many properties of the symmetric matrices. First
of all, $\G$  can be transformed to a symmetric matrix simply by
considering $D^{1/2} \G D^{-1/2}$. It follows that any  eigenvector
of $D^{1/2}\G D^{-1/2}$ can be transformed into a left (right)
eigenvector of $\G$, once it is multiplied by $D^{1/2}$
($D^{-1/2}$), respectively. Henceforth, the left and right
eigenvalues of $\G$ are the same, and they are real.

Furthermore,  since  $ \norm{D^{-1} A}_\infty\leq 1$, and  $\G$ is the average of $D^{-1} A$,
and the identity matrix, we must have $\lambda_{\min}(\G)\geq 0$ and $\lambda_{\max}(\G) \leq 1$.
Thus  $D^{1/2}\G D^{-1/2}$ is a positive semidefinite matrix, symmetric matrix,
whose largest eigenvalue is at most 1.

\subsection{The Evolving Set Process}
The {\em evolving set process} is a markov chain on the subsets of the vertex set $V$. The process together with the closely related {\em volume biased evolving set process} is introduced in the work of Diaconis and Fill \cite{DF90} as
the {\em strong stationary dual} of a random walk. Morris and Peres \cite{MP03} use it to upper-bound the mixing time of random walks in terms of isoperimetric properties.

Given a subset $S_0\subseteq V$, the next subset $S_{1}$ is chosen as
follows: first we choose a threshold $R\in [0,1]$ uniformly at random. Then, we let
$$ S_{1}:= \{ u: \cPP{u}{X_1\in S_0} \geq R\}. $$
The transition kernel of the evolving set process is defined as, $\textup{K}(S,S') = \P{S_1=S' | S_0=S}$.
It follows that $\emptyset$, and $V$ are the absorbing states of this markov chain, and the rest of the states are transient.
Morris and Peres \cite{MP03}
defined the {\em growth gauge} $\psi(S)$ of a set $S$ as follows:
$$ \psi(S):=1-\E{\sqrt{\frac{\mu(S_1)}{\mu(S)}}~\middle|~ S_0=S}.$$
They showed that
\begin{proposition}[{Morris,Peres \cite{MP03}}]
\label{prop:growthgauge}
For any set $S\subseteq V$, $\psi(S)\geq \phi(S)^2/8$.
\end{proposition}

The {\em volume-biased evolving set process} is a special case of the evolving set process where the markov chain is conditioned to be absorbed in $V$.
In particular, the transition kernel is defined of a volume biased ESP as follows:
$$ {\bf K}(S,S') = \frac{\mu(S')}{\mu(S)}{\textup{K}}(S,S').$$
Given a state $S_0$, we write $\hPP{S_0}{.}:=\hP{.|S_0}$ to denote the probability measure of the volume biased ESP started at state $S_0$, and we use $\hEE{S_0}{.}:=\hE{.|S_0}$ for the expectation.

Andersen and Peres used the volume biased ESP as a local graph clustering algorithm \cite{AP09}. They show that for any non-expanding set $U$, if we run the volume biased ESP from a randomly chosen vertex of $U$, with a constant probability, there is a set in the sample path of expansion $O(\sqrt{\phi(U)\log n})$, and volume at most $2\mu(U)$.
As a part of their proof, they designed an efficient   simulation of the volume biased ESP,
called $\gensam$. They prove the following theorem,
\begin{theorem}[{Andersen, Peres \cite[Theorems 3,4]{AP09}}]
\label{lem:espsimulate}
There is an algorithm, $\gensam$, that simulates the volume biased ESP such that for any vertex $v\in V$, any sample path $(S_0=\{v\},\ldots,S_\tau)$, is generated  with probability
$\hPP{v}{S_0,\ldots,S_\tau}$.
Furthermore, for a stopping time $\tau$ that is bounded above by $T$,
let $W(\tau)$ be the time complexity of $\gensam$ if it is run up to time $\tau$.
Then, the expected work per volume ratio of the algorithm is
$$  \hEE{v}{\frac{W(\tau)}{\mu(S_\tau)}} = O(T^{1/2}\log^{3/2}\mu(V)).$$
\end{theorem}

\def\U{S}
\section{Upper Bounds on the Escaping Probability of Random Walks}
\label{sec:escprob}
In this section we establish strong results on the escaping probability of the random walks.
Spielman and Teng \cite{ST08} show that for any set $\U\subseteq V$, $t>0$, the random walk started at a randomly (proportional to degree) chosen vertex of $\U$, remain
in $\U$ for $t$ steps with probability at least $1-t\phi(\U)/2$.
We strengthen this result, by improving the lower bound to $(1-\phi(\U)/2)^t$,
\begin{proposition}
\label{prop:escprob}
For any set $\U\subseteq V$, and integer $t>0$,
\begin{equation}
\label{eq:remlowerbound}
 \EE{v\sim \bpsi_\U}{\rem(v,t,\U)} \geq  \left(1-\frac{\phi(\U)}{2}\right)\EE{v\sim \bpsi_\U}{\rem(v,t-1,\U)} \geq \ldots \geq \left(1-\frac{\phi(\U)}{2}\right)^t.
 \end{equation}
Furthermore, there is a subset $\U^t \subseteq \U$, such that $\mu(\U^t)\geq \mu(\U)/2$, and for all $v\in \U^t$
\begin{equation}
\label{eq:halfremlowerbound}
 \rem(v,t,\U)  \gtrsim \left(1-\frac{3\phi(\U)}{2}\right)^t.
 \end{equation}
\end{proposition}
We remark
that the second statement does not follow from a simple application
of the Markov Inequality to the first statement, as this is the case
in \cite{ST08}. Whence, here both of the results incorporate
non-trivial spectral arguments.

As a corollary, we  prove strong lower bounds on the uniform mixing time of random walks in section \ref{sec:mixingtime}.
In the rest of this section we prove \autoref{prop:escprob}.
We start by proving \eqref{eq:remlowerbound}.

Using equation \eqref{eq:remprob}, and a simple induction on $t$,  \eqref{eq:remlowerbound} is equivalent to the following equation:
\begin{equation}
\label{eq:remonelevelbound} \tp{\bpsi_\U}(I_\U \G I_\U)^t\b1_\U \geq
(1-\phi(\U)/2) \tp{\bpsi}(I_\U \G I_\U)^{t-1}\b1_\U.
\end{equation}

Let $P := D^{1/2} I_\U \G I_\U D^{-1/2}$. First we show that
\eqref{eq:remonelevelbound} is equivalent to the  following
equation:
\begin{equation}
\label{eq:normalizedescape}
\tp{ \sqrt {\bpsi_\U}} P^t \sqrt{\bpsi_\U}  \geq \left(\tp{\sqrt{\bpsi_\U}} P \sqrt{\bpsi_\U}\right)\left(\tp{\sqrt{\bpsi_\U}} P^{t-1}\sqrt{\bpsi_\U}\right).
 \end{equation}
Then, we use  \autoref{lem:matrixnormmon} that shows
the above equation holds for any symmetric positive semidefinite matrix  $P$, and any norm one vector $\bx=\sqrt{\bpsi_\U}$.
First observe that by the definition of $P$, for any
$t>0$,
\begin{equation}
\label{eq:leftPsqrt} \tp{\bpsi_\U} (I_\U \G I_\U)^t \b1_\U =
\tp{\bpsi_\U} D^{-1/2} P^t D^{1/2} \b1_\U  = \tp{\sqrt{\bpsi_\U}}
P^t \sqrt{\bpsi_\U}
\end{equation}
On the other hand,
\begin{eqnarray} \tp{\bpsi_\U} (I_\U \G I_\U) \b1_\U &=& \frac12\tp{\bpsi_\U} (D^{-1}A + I) \b1_\U \nonumber\\
&=& \frac{1}{2\mu(\U)} (\tp{\b1_\U} A \b1_\U)  + \frac12 \tp{\bpsi_\U} \b1_\U\nonumber\\
&=& \frac{1}{2\mu(\U)} 2|E(\U)| + \frac12\nonumber\\
&=& 1-\phi(\U)/2 \label{eq:rightPsqrt}.
\end{eqnarray}
Equation \eqref{eq:normalizedescape} is derived simply  from
equation \eqref{eq:remonelevelbound}, by putting
\eqref{eq:leftPsqrt}, \eqref{eq:rightPsqrt} together. Next we prove
equation \eqref{eq:normalizedescape} using
\autoref{lem:matrixnormmon}. First observe that $\sqrt{\bpsi_\U}$ is
a norm one vector. On the other hand, by definition $P = \frac12(
D^{-1/2} I_\U A I_\U D^{-1/2} + D^{-1/2} I_\U D^{-1/2}) $ is a symmetric matrix.

It remains to show that $P$ is positive semidefinite. This follows
from the same reason that $D^{1/2} \G D^{-1/2}$ is positive semidefinite.
In particular, since eigenvectors of $P$ can be transformed to the eigenvectors of $I_\U \G I_\U$, 
 the  eigenvalues of $I_\U \G I_\U$ are the same as the eigenvalues of $P$.
Finally,  since  $\norm{I_\U D^{-1} A I_\U}_\infty \leq 1$, and
$I_\U \G I_\U$ is the average of $I_\U D^{-1}AI_\U$ and $I_\U$, we
must have $\lambda_{\min}(I_\U \G I_\U)\geq 0$, and
$\lambda_{max}(I_\U \G I_\U)\leq 1$. Thus $P$ is positive
semidefinite. Now, \eqref{eq:normalizedescape} simply follows from
\autoref{lem:matrixnormmon}. This completes the proof of
\eqref{eq:remlowerbound}

It remains to prove \eqref{eq:halfremlowerbound}.
We prove it by  showing that for any set $X\subseteq \U$, of volume $\mu(X)\geq \mu(\U)/2$, the random walk started at a randomly (proportional to degree) chosen vertex of $X$, remains in $X$ (and $\U$), with probability at least $\frac{1}{200} (1-3\phi(\U)/2)^t$,
\begin{equation}
\label{eq:Xremprob} \EE{v\sim \bpsi_X}{\rem(v,t,X)} =  \tp{\bpsi_{X}}(I_\U \G I_\U)^t\b1_X \geq
\frac{1}{200}\left(1-\frac{3\phi(\U)}{2}\right)^t.
\end{equation}
Therefore, in any such set $X$, there is a vertex that satisfy
\eqref{eq:halfremlowerbound}, hence the  volume of the set of
vertices that satisfy \eqref{eq:halfremlowerbound} is at least half
of $\mu(\U)$.

Using equations \eqref{eq:leftPsqrt} and \eqref{eq:rightPsqrt},
\eqref{eq:Xremprob}  is equivalent to the following equation,
\begin{equation}
\tp{\sqrt{\bpsi_X}} P^t \sqrt{\bpsi_X} \geq \frac{1}{200} \left(
3\tp{\sqrt{\bpsi_\U}} P \sqrt{\bpsi_\U} - 2\right)^t.
\end{equation}
We prove the above equation using \autoref{lem:twovectors}.
Let $Y=\U\setminus X$, and define
\begin{eqnarray}
\bx&:=&I_X\sqrt{\bpsi_\U} = \sqrt{\mu(X)\bpsi_X/\mu(\U)}\label{eq:xpsix}\\
\by&:=&I_Y\sqrt{\bpsi_\U} = \sqrt{\mu(Y)\bpsi_Y/\mu(\U)}\nonumber
\end{eqnarray}
Since $X\cap Y=\emptyset$, $\langle \bx,\by\rangle =0$, and $\norm{\bx+\by}=\norm{\sqrt{\bpsi_\U}}=1$. Furthermore, since $\mu(X) \geq \mu(\U)/2 \geq \mu(Y)$, $\norm{\bx}\geq \norm{\by}$.
Therefore, $P$, $\bx,\by$ satisfy the requirements of \autoref{lem:twovectors}.
Finally, since $ \tp{\sqrt{\bpsi_X}} P^t\sqrt{\bpsi_X} \geq \tp{\bx}P^t \bx$,  \eqref{eq:Xremprob} follows from \autoref{lem:twovectors}. 
This completes the proof of Proposition \ref{prop:escprob}.

\begin{lemma}
\label{lem:matrixnormmon}
Let $P \in \mathbb{R}^{n\times n}$ be a symmetric positive semidefinite matrix.
Then, for any  $\bx\in \mathbb{R}^n$ of norm $\norm{x}=1$, and  integer $t> 0$,
$$ \tp{\bx} P^t \bx \geq \left(\tp{\bx} P^{t-1}\bx\right)\left(\tp{\bx} P \bx\right) \geq \ldots \geq \left(\tp{\bx} P\bx\right)^t.$$
\end{lemma}
\begin{proof}
Since all of the inequalities in lemma's statement follows from the
first inequality, we only prove the first inequality. Let $\bv_1,
\bv_2,\ldots,\bv_n$ be the set of orthonormal eigenvectors of $P$
with the corresponding eigenvalues $\lambda_1,\lambda_2, \ldots,
\lambda_n$. For any $k\geq 1$, we have
\begin{eqnarray}
\tp{\bx} P^k \bx &=& \left(\sum_{i=1}^n \langle \bx, \bv_i\rangle \bv_i \right) P^k \left( \sum_{i=1}^n \langle \bx, \bv_i\rangle \bv_i\right) \nonumber\\
&=& \left(\sum_{i=1}^n \langle \bx,  \bv_i\rangle \lambda_i^k \bv_i \right) \cdot \left( \sum_{i=1}^n \langle \bx, \bv_i\rangle \bv_i\right) \nonumber\\
&=& \sum_{i=1}^n \langle \bx,  \bv_i\rangle^2 \lambda_i^k. \label{eq:xpxexpanding}
\end{eqnarray}
On the other hand, since $\{\bv_1,\ldots,\bv_n\}$ is an orthornormal system, we have
$$ \sum_{i=1}^n \langle \bx, \bv_i\rangle^2=\norm{\bx}^2=1.$$
For any $k>0$, Let $f_k(\lambda) = \lambda^k$; it follows that,
$$\sum_i \langle \bx, \bv_i\rangle^2 \lambda_i^k = \EE{\Lambda\sim\cD}{f_k(\Lambda)},$$
where $\PP{\Lambda\sim\cD}{\Lambda=\lambda_i} = \langle \bx, \bv_i\rangle^2$. Using equation \eqref{eq:xpxexpanding} we may rewrite the lemma's statement as follows:
$$ \EE{\cD}{f_{t-1}(\Lambda)f_{1}(\Lambda)} \geq \EE{\cD}{f_{t-1}(\Lambda)} \EE{\cD}{f_1(\Lambda)}$$
Since $P$ is positive semidefinite, $\lambda_{\min}(P)\geq 0$. Thus, for all $t>0$, the function $f_t(.)$ is increasing  in the support of $\cD$.
The above inequality follows from the Chebyshev's sum inequality.
\end{proof}

\begin{lemma}
\label{lem:twovectors}
Let $P\in \mathbb{R}^{n\times n}$ be a symmetric positive semidefinite  matrix such that $\lambda_{\max}(P)\leq 1$, and $\bx,\by \in \mathbb{R}^n$ such that $\langle \bx,\by\rangle =0$, $\norm{\bx+\by}=1$,
and $\norm{\bx} \geq \norm{\by}$. Then, for any integer $t>0$,
$$ \tp{\bx} P^t \bx \geq \frac{1}{200} \left(3\tp{(\bx+\by)} P (\bx+\by)-2\right)^t$$
\end{lemma}
\begin{proof}
Let $\bz:=\bx+\by$. Since $\bx$ is orthogonal to $\by$, we have $\norm{\by}^2\leq 1/2 \leq \norm{\bx}^2$.
Let $\bv_1, \bv_2,\ldots,\bv_n$ be the set of orthonormal eigenvectors of $P$  with the corresponding eigenvalues $\lambda_1, \lambda_2, \ldots, \lambda_n$.
Let $\alpha>0$ be a constant that will be fixed later in the proof. Define $B:=\{i: |\langle \bx,\bv_i\rangle| \geq \alpha |\langle \by,\bv_i\rangle|\}$.
First observe that,
\begin{equation}
\label{eq:xpxlowerbound}
 \tp{\bx} P^t \bx = \sum_{i=1}^n \langle \bx,\bv_i\rangle^2 \lambda_i^t  \geq \sum_{i\in B} \langle \bx,\bv_i\rangle^2 \lambda_i^t
\geq \frac{1}{(1+1/\alpha)^2}\sum_{i\in B} \langle \bz,\bv_i\rangle^2 \lambda_i^t,
\end{equation}
where the  equality follows from equation \eqref{eq:xpxexpanding}, the first inequality uses $\lambda_{\min}(P)\geq 0$, and the last inequality follows from the definition of $B$, that is for any $i\in B$, $\langle \bx,\bv_i\rangle^2 \geq
(\langle \bz,\bv_i\rangle/(1+1/\alpha))^2$.
Let $L:=\sum_{i\in B} \langle \bz,\bv_i\rangle^2$. Then, since $\lambda_{\min}(P)\geq 0$, by Jensen's inequality,
\begin{eqnarray}
\frac{1}{L}\sum_{i\in B} \langle \bz,\bv_i\rangle^2 \lambda_i^t &\geq& \left(\frac{1}{L}\sum_{i\in B} \langle \bz,\bv_i\rangle^2 \lambda_i\right)^t\nonumber\\
&\geq&  \left(\frac{\sum_{i=1}^n \langle \bz,\bv_i\rangle^2 \lambda_i - (1-L)}{L}\right)^t\nonumber\\
&\geq & \left( 1 - \frac{1-\tp{\bz}P\bz}{(1-\alpha^2-2\alpha)/2}\right)^t,\label{eq:Bbound}
\end{eqnarray}
where the second inequality follows by the assumptions that $\lambda_{\max}(P)\leq 1$, and that $\norm{\bz}=1$,
and the last inequality follows from the fact that $\tp{\bz} P \bz \leq 1$, and that
$$ L=\sum_{i\in B} \langle \bz,\bv_i\rangle^2  = 1- \sum_{i\notin B} \langle \bz,\bv_i\rangle^2  \geq1- (1+\alpha)^2 \norm{\by}^2 \geq \frac{1-\alpha^2-2\alpha}{2}.$$
Putting equations \eqref{eq:xpxlowerbound} and \eqref{eq:Bbound}, and letting $\alpha=0.154$ we get,
$$ \tp{\bx} P^t \bx \geq \frac{1-\alpha^2-2\alpha}{2(1+1/\alpha)^2}  \left( 1 - \frac{1-\tp{\bz}P\bz}{(1-\alpha^2-2\alpha)/2}\right)^t \geq \frac{1}{200}\left( 3\tp{\bz}P\bz-2)\right)^t$$
\end{proof}

\section{Approximating the Expansion Profile}
\label{sec:expprofile}
In this section we use the machinery developed in the works of Lovasz and Simonovits \cite{LS90,LS93} to prove  \autoref{thm:sse}, \autoref{thm:abs}. We start by introducing some notations.

Let $\bp$ be a probability distribution vector on the vertices of $V$, and let $\sigma(.)$ be the permutation of the vertices that is decreasing with respect to $p(v)/d(v)$, and
breaking ties lexicographically.
That is, suppose
$$  \frac{p(\sigma(1))}{d(\sigma(1))} \geq \ldots \geq \frac{p(\sigma(i))}{d(\sigma(i))} \geq \ldots \geq \frac{p(\sigma(n))}{d(\sigma(n))}.$$
We use $T_i(\bp):=\{\sigma(1),\ldots,\sigma(i)\}$ to denote  the {\em threshold set} of the first $i$ vertices.
Following Spielman, Teng \cite{ST08} (c.f. Lovasz, Simonovits \cite{LS90}), we use the following potential function,
\begin{equation}
\label{eq:Idef} I(\bp,x):= \max_{\substack{w\in [0,1]^n \\ \sum w(v)d(v)=x}} \sum_{v\in V} w(v) p(v).
\end{equation}
Observe that for $I(\bp,x)$ is a non-decreasing  piecewise linear concave function of $x$, that is  $I(\bp,x) = p(T_j(\bp))$, for $x=\mu(T_j(\bp))$, and is linear in other values of $x$.
We use $I(\bp,x)$ as a potential function to measure the distance of the distribution $\bp$  from the stationary distribution $\bpsi$.

We find the small non-expanding set in \autoref{thm:sse}, by
running $\SSE(\sqrt{2\phi/\eps}, \eps\ln{\mu(V)}/\phi)$.
The algorithm simply returns the smallest non-expanding set among the threshold sets of the rows of $\G^t$, for $t=O(\eps \log\mu(V)/\phi)$.
The details are described in Algorithm~\autoref{alg:sse}.

\begin{algorithm}[htb]
\caption{$\SSE(\phi,\eps)$ }
\label{alg:sse}
\begin{algorithmic}
\STATE Let ${\cal T}$ be the family of all threshold sets $T_i(\tp{\b1_v} \G^t)$, for any vertex $v\in V$, and $1\leq t\leq \eps \log\mu(V)/\phi/$, with conductance at most $\sqrt{2\phi/\eps}$.
\RETURN the set with minimum volume in ${\cal T}$.
\end{algorithmic}
\end{algorithm}

If none of the sets $T_i(\tp{\b1_{v}}\G^t)$ is a  non-expanding set, then
Lovasz and Siminovits \cite{LS90,LS93} prove that the curve $I(\tp{\b1_{v}}\G^{t},x)$ lies far below
$I(\tp{\b1_{v}}\G^{t-1},x)$.
This is quantified in the following lemma:
\begin{lemma}[{Lovasz, Simonovits  \cite[Lemma 1.3]{LS93}}]
\label{lem:ls}
Let $\G$ be a transition probability matrix of a lazy random walk on a graph. For any probability distribution vector $\bp$ on $V$, if $\phi(T_i(\tp{\bp} \G))\geq \Phi$, then for $x=\mu(T_i(\tp{\bp} \G))$,
$$ I(\tp{\bp} \G, x) \leq \frac12 \left(I(\bp, x-2\Phi\min(x,2m-x) + I(\bp, x+2\Phi\min(x,2m-x)\right).$$
\end{lemma}

By repeated application of the above lemma, Lovasz and Simonovits \cite{LS90}
argue that, if all of the sets $T_i(\tp{\b1_{v}}\G^t)$ are expanding,
then $I(\tp{\b1_{v}}\G^t,.)$ approaches the straight line.
In the next lemma we show that, if all of the small threshold sets (i.e., $\mu(T_i(\tp{\b1_{v}} \G^t))\leq \Gamma$), are expanding, then $I(\tp{\b1_{v}} \G^t,.)$ approaches the curve $x/\Gamma$.
\begin{lemma}
\label{lem:recursels}
For any vertex $v\in V$,  $t\geq 0$, and $0\leq \Gamma \leq m$, $0\leq \Phi\leq 1/2$, if for all $t \leq T$, all of threshold sets $T_i(\tp{\b1_{v}}\G^t)$ of volume at most $\Gamma$,
has expansion at least $\Phi$, then for any $0\leq t\leq T$,
$$ I(\tp{\b1_{v}} \G^t, x) \leq \frac{x}{\Gamma} + \sqrt{x/\mu(v)} \left(1-\frac{\Phi^2}{2}\right)^t.$$
\end{lemma}
\begin{proof}
We prove by induction. The lemma trivially holds for $t=0$.
This is because the LHS is $x/\mu(v)$ for $0\leq x\leq \mu(v)$ and $1$ for larger values of $x$, while the RHS is $\sqrt{x/\mu(v)}$ for all $x\geq0$.
Next, we prove the lemma's statement holds for $t$, assuming that it holds for $t-1$.
Let $\bp:=\tp{\b1_{v}} \G^{t-1}$. First of all, since $I(\tp{\bp} \G,.)$ is a piecewise-linear concave function of $x$, it is sufficient to prove the statement for values of $x=\mu(T_i(\tp{\bp} \G))$.
For $x\geq \Gamma$, the statement holds trivially, because the RHS is at least 1, while the LHS is less than or equal to 1.
Now, suppose $x<\Gamma$, and $x=\mu(T_i(\tp{\bp}\G))$.  Using \autoref{lem:ls}, we have
\begin{eqnarray*}
I(\tp{\bp} \G,x) &\leq& \frac12 \left\{I(\bp, x-2\Phi x) + I(\bp, x+2\Phi x)\right\}\\
&\leq &\frac12\left\{ \frac{2x}{\Gamma} + \sqrt{x/\mu(v)}\left(1-\frac{\Phi^2}{2}\right)^{t-1} \left(\sqrt{1-2\Phi} + \sqrt{1+2\Phi}\right)\right\} \\
&\leq & \frac{x}{\Gamma} + \sqrt{x/\mu(v)} \left(1-\frac{\Phi^2}{2}\right)^t,
\end{eqnarray*}
where the first inequality uses the assumption that $x<\Gamma\leq m$, the second inequality uses the induction hypothesis,
and the last inequality uses the inequality
$$ \frac12\left(\sqrt{1-2\Phi} + \sqrt{1+2\Phi}\right) \leq 1-\frac{\Phi^2}{2},$$
holds for any $\Phi\leq 1/2$.
\end{proof}

Now we are ready to prove \autoref{thm:sse}: using \autoref{prop:escprob} we show that $I(\tp{\b1_{v}}\G^t)$ does not converge to $x/\Gamma$, for $t\approx \log\gamma/\phi$. Therefore, by the previous lemma, at least one of the small threshold sets is non-expanding.

\mainsse*
\begin{proof}
Let, $\gamma=\mu(U)$, $T=\eps \ln{\gamma}/\phi$, $\Gamma=2\gamma^{1+\eps}$, and $\Phi=\sqrt{2\phi/\eps}$.
We show that $\SSE(\phi,\eps)$ returns a set of volume at most $\Gamma$, and conductance at most $\Phi$.
Wlog we may assume that  $\Gamma < m$, otherwise the statement is trivial.
 We prove by contradiction; assume that
the output of the algorithm has volume larger than $\Gamma$, we show that \autoref{lem:recursels} and \autoref{prop:escprob} can not hold simultaneously.
First of all, by \autoref{prop:escprob}, there exists a vertex $u\in U$, such that
$$ \rem(u,t,U) \geq \left(1-\frac{\phi(U)}{2}\right)^t \geq \left(1-\frac{\phi}{2}\right)^{\frac{\eps\ln{\gamma}}{\phi}} \geq e^{-\eps \ln{\gamma}}  \geq \gamma^{-\eps}.$$
Let $\bp=\tp{\b1_{u}} \G^T$. Let $w(v)=1$, for all $v\in U$, and $w(v)=0$ for the rest of the vertices. By equation \eqref{eq:Idef}, we have
$$ I(\bp, \gamma) \geq \sum_{v\in V} w(v) p(u) = \sum_{v\in U} p(v) \geq \rem(u,t,U) \geq \gamma^{-\eps}. $$
On the other hand, by  \autoref{lem:recursels}, we have:
\begin{eqnarray*}
 I(\bp, \gamma) \leq \frac{\gamma}{\Gamma} + \sqrt{\gamma} \left(1-\frac{\Phi^2}{2}\right)^t
\leq  \frac{1}{2\gamma^{\eps}} + \sqrt{\gamma} (1-\frac{\phi}{\eps})^{\frac{\eps\ln{\gamma}}{\phi}} < \gamma^{-\eps},
\end{eqnarray*}
which is a contradiction.
The last inequality uses the fact that $\eps<1/2$.
\end{proof}

Next we show that using \autoref{lem:recursels} we can provide a simpler, and yet stronger proof of the result of Arora, Barak and Steurer \cite{ABS10}.
\mainabs*
 \begin{proof}
Wlog we assume $\phi\leq 1/2$, $\eta < \phi$, and $n^{\eta/\phi}>4$. Let $T=\eps\ln{n}/\phi$, $\Gamma=4\mu(V)n^{-\eta/\phi}$, and $\Phi=\sqrt{2\phi/\eps}$.
 We show that $\SSE(\Phi,T)$ finds a set of volume $\Gamma$ and conductance at most $\Phi$.
 We  prove by contradiction; suppose that $\SSE$ does not find such a set.
 Since $\G=\frac12(D^{-1}A + I)$, $\rank_{1-\eta/2}(\G) \geq n^{(1+\eps)\eta/\phi}$. Therefore, by the next claim, there is a vertex $u$ such that,
 $$ \tp{\b1_u}\G^t\b1_u \geq \max\left\{\frac{1}{2n},\frac{\mu(u)}{2\mu(V)}\right\} n^{(1+\eps)\eta/\phi} (1-\eta/2)^{\eps\ln{n}/\phi} \geq \max\left\{\frac{1}{2n},\frac{\mu(u)}{2\mu(V)}\right\}n^{\eta/\phi},$$
 Let $\bp=\tp{\b1_{u}} \G^T$, $x=\mu(u)$, $w(v)=1$, for $v=u$ and $w(v)=0$ for the rest of the vertices. By equation \eqref{eq:Idef}, we have
 $$ I(\bp, x) \geq \sum_{v\in V} w(v) p(v) = p(u)  \geq \max\left\{\frac{1}{2n},\frac{\mu(u)}{2\mu(V)}\right\} n^{\eta/\phi}.$$
 But, by \autoref{lem:recursels}, we have
 $$ I(\bp, x) \leq \frac{x}{\Gamma} + \sqrt{x/\mu(u)} \left(1-\frac{\Phi^2}{2}\right)^T \leq \frac{\mu(u)}{4\mu(V)} n^{\eta/\phi} + \frac{1}{n},$$
 which is a contradiction, since $n^{\eta/\phi} > 4$.
 \begin{claim}
 For any graph $G$, if $\rank_{1-\eta}(\G) \geq r$, then there is a vertex $u\in V$, such that
 $$ \tp{\b1_{u}} \G^t \b1_{u} \geq \max\left\{\frac{1}{2n},\frac{\mu(u)}{2\mu(V)}\right\}\cdot r\cdot (1-\eta)^t . $$
 \end{claim}
 \begin{proof}
 Let $0\leq \lambda_1,\ldots,\lambda_n\leq 1$ be the eigenvalues of $\G$.
 We use the trace formula,
 $$\sum_{v\in V} \b1_{v} \G^t \b1_{v} = \tr(\G^t) = \sum_{i=1}^n \lambda_i^t \geq r\cdot(1-\eta)^t.$$
 Now, let $U_1:=\{v: \b1_{v}\G^t\b1_{v} < r(1-\eta)^t/2n\}$, and $U_2:=\{v: \bv1_{v} \G^t\b1_{v} < \frac{\mu(v)}{2\mu(V)}r(1-\eta)^t\}$.
 It follows that,
 $$ \sum_{v\in U_1} \b1_{v} \G^t\b1_{v} + \sum_{v\in U_2} \b1_{v} \G^t\b1_{v} < r\cdot(1-\eta)^t \left(\frac{|U_1|}{2n} + \frac{\mu(U_2)}{2\mu(V)}\right) \leq r\cdot(1-\eta)^t.$$
 Therefore, there is a vertex $u\notin U_1, U_2$ that satisfies claim's statement.
 \end{proof}
 \end{proof}

\section{Almost Optimal Local Graph Clustering}
In this section we use the volume biased ESP to design a local graph clustering algorithm with a worst case guarantee on the conductance of output set that is independent of the size of $G$. 

Let $(S_0,S_1,\ldots,S_\tau)$ be a sample path of the volume biased ESP, for a stopping time $\tau$. Andersen and Peres show that with a constant probability the conductance of at least one of the sets in the sample path is at most $O(\sqrt{\frac{1}{\tau}\log\mu(S_\tau)})$,
\begin{lemma}[{\cite[Lemma1,Corollar 1]{AP09}}]
For any starting set $S_0$, and any stopping time $\tau$, and $\alpha>0$,
$$ \hPP{S_0}{\sum_{i=1}^\tau \phi^2(S_i) \leq 4\alpha \ln\frac{\mu(S_\tau)}{\mu(S_0}} \geq 1-\frac{1}{\alpha}.$$
\end{lemma}
Here, we strengthen the above result, and we show the event occurs with much higher probability.
In particular, we show the with probability at least $1-1/\alpha$, the conductance of at least one of the sets in the sample path  is  at most $O(\sqrt{\frac{1}{\tau}\log(\alpha \cdot \mu(S_\tau))})$.
\begin{lemma}
\label{lem:concentration}
For any starting set $S_0\subseteq V$, and any stopping time $\tau$, and $\alpha\geq 0$,
$$ \hPP{S_0}{\sum_{i=1}^\tau \phi^2(S_i) \leq 8\left(\ln{\alpha} + \ln{\frac{\mu(S_\tau)}{\mu(S_0)}}\right)} \geq 1-\frac{1}{\alpha} $$
\end{lemma}
\begin{proof}
Let
\begin{eqnarray*}
M_t&:=& \frac{\sqrt{\mu(S_0)}}{\sqrt{\mu(S_t)}} \prod_{i=0}^{t-1} \frac{1}{1-\psi(S_i)}
\end{eqnarray*}
Andersen, Peres \cite[Lemma 1]{AP09} show that $M_t$ is a martingale in the volume biased ESP. It follows  from the optional sampling theorem \cite{Williams91} that $\hE{M_\tau} = M_0 = 1$.
Thus, by the Markov inequality,  for any $\alpha>0$, we have
$$ \hP{M_\tau \leq \alpha} \geq 1-\frac{1}{\alpha}$$
By taking logarithm, from both sides of the event in the above equation we obtain,
\begin{equation}\label{eq:logmtaumarkov} \hP{\ln{M_\tau} \leq \ln{\alpha}} \geq 1-\frac{1}{\alpha}
\end{equation}
On the other hand, by the definition of $M_\tau$,
\begin{eqnarray}
\ln{M_\tau} &=& \frac12 \ln\frac{\mu(S_0)}{\mu(S_\tau)} + \ln{\prod_{i=0}^{\tau-1} \frac{1}{1-\psi(S_i)}} \nonumber \\
&\geq & \frac12 \ln\frac{\mu(S_0)}{\mu(S_\tau)} + \sum_{i=0}^{\tau-1} \psi(S_i) \nonumber \\
&\geq & \frac12 \ln\frac{\mu(S_0)}{\mu(S_\tau)} + \frac18\sum_{i=0}^{\tau-1} \phi^2(S_i),\label{eq:logmtaulowerbound}
\end{eqnarray}
where the first inequality follows by the fact that $1/(1-\psi(S_i)) \geq e^{\psi(S_i)}$, and the last inequality follows by \autoref{prop:growthgauge}.
Putting \eqref{eq:logmtaumarkov} and \eqref{eq:logmtaulowerbound} together proves the lemma.
\end{proof}

The previous lemma shows that for any $\gamma,\phi>0$, if we can run the process for $T\approx \eps\log\gamma/\phi$ steps without observing a set  larger than $\gamma^{O(1)}$, then, with  probability $1-1/\gamma$, one of the sets in the sample path must have an expansion of $O(\sqrt{\phi/\eps})$, which is what we are looking for.
Next we use the following lemma by Andersen and Peres, together with \autoref{prop:escprob}, to show that event occurs with  some non-zero probability. That is, for any $\eps<1$, with probability at least $\approx\gamma^{-\eps}$, the volume of all sets in the sample path of the process are at most $O(\gamma^{1+\eps})$. Then, by the union bound we can argue both event occur with probability at least $\Omega(\gamma^{-\eps})$.

\begin{lemma}[{Andersen, Peres \cite[Lemma 2]{AP09}}]
\label{lem:dfcoupling}
For any set $U\subset V$, $v\in U$, and integer $T>0$, the following holds,
$$ \hPP{v}{\max_{t\leq T} \frac{\mu(S_t\setminus U)}{\mu(S_t)} > \beta \esc(v,T,U)} < \frac{1}{\beta}, ~~\forall \beta>0$$
\end{lemma}

\begin{lemma}
\label{lem:espsuccessprob}
Given $\gamma>0$, $0<\phi<1/4$, $0<\eps<1$, such that $G$ has a set $U\subset V$ of volume $\mu(U) \leq \gamma$, and conductance $\phi(U)\leq \phi$.
Let $T=\eps\ln{\gamma}/3\phi$. There is a constant $c>0$, and a subset $U^T \subseteq U$ of volume $\mu(U^T)\geq \mu(U)/2$ such that for any $v\in U^T$,
with probability at least $c\gamma^{-\eps}/8$, a sample path $(S_1,S_2,\ldots,S_T)$ of the volume biased ESP started from $S_0=\{v\}$ satisfies the following,
\begin{enumerate}[i)]
\item For some $t\in [0,T]$, $\phi(S_t) \leq O(\sqrt{100(1-\ln c)\phi/\eps})$,
\item For all $t\in [0,T]$, $\mu(S_t\cap U) \geq c\gamma^{-\eps}\mu(S_t)/2$, and henceforth, $\mu(S_t)\leq 2\gamma^{1+\eps}/c$.
\end{enumerate}
\end{lemma}
\begin{proof}
First of all, we let $U^T$ be the set of vertices $v\in U$ such that
$$ \rem(v,T,U) \geq c\left(1-\frac{3\phi(U)}{2}\right)^T.$$
 By  \autoref{prop:escprob}, there exists a constant $c>0$ such that $\mu(U^T)\geq \mu(U)/2$.
In the rest of the proof let  $v$ be a vertex in $U^T$. We have,
$$ \esc(v,T,U) 
\leq 1-c\left(1-\frac{3\phi(U)}{2}\right)^T \leq 1-c\left(1-\frac{3\phi}{2}\right)^{\frac{\eps\ln{\gamma}}{3\phi}}  \leq 1-c\gamma^{-\eps}$$
Now, let $\beta:=1+c\gamma^{-\eps}/2$. By \autoref{lem:dfcoupling}, we have
$$ \hPP{v}{\max_{t\leq T} \frac{\mu(S_t\setminus U)}{\mu(S_t)} \leq \beta \esc(v,T,U) \leq 1-\frac{c\gamma^{-\eps}}{2}} \geq 1-\frac{1}{\beta} \geq  \frac{c\gamma^{-\eps}}{4}$$
Since for any $S\subset V$, $\mu(S\setminus U) + \mu(S\cap U)=\mu(S)$, we have
$$ \hPP{v}{\min_{t\leq T} \frac{\mu(S_t\cap U)}{\mu(S_t)} \geq \frac{c\gamma^{-\eps}}{2}} \geq \frac{c\gamma^{-\eps}}{4}$$
On the other hand, let $\alpha:=\gamma$. By \autoref{lem:concentration}, with probability $1-1/\gamma$, for some $t\in [0,T]$,
$$ \phi^2(S_t) \leq \frac{1}{T} \sum_{t=0}^T \phi^2(S_t) \leq \frac{8(\ln{\gamma} + \ln{\mu(S_T)})}{T}. $$
Therefore, since $\eps<1$, by the union bound we have
$$\hPP{v}{\min_{t\leq T} \frac{\mu(S_t\cap U)}{\mu(S_t)}\geq\frac{c\gamma^{-\eps}}{2} ~~\bigwedge~~ \exists~t:~\phi(S_t) \leq  \sqrt{\frac{8(\ln{\gamma} + \ln{\mu(S_T)})}{T}}} \geq \frac{c\gamma^{-\eps}}{8}$$
Finally, since for any set $S\subseteq V$, $\mu(S\cap U)\leq \mu(U)\leq \gamma$, in the above  event, $\mu(S_T) \leq \frac{2\gamma^{1+\eps}}{c}$. Therefore,
$$ \phi(S_t)\leq \sqrt{\frac{8(\ln{\gamma}+\ln(2\gamma^{1+\eps}/c))}{T}} \leq \sqrt{\frac{100(1-\ln{c}) \phi}{\eps}},$$
which completes the proof.
\end{proof}

To prove  \autoref{thm:ssefast}, we can simply run $\gamma^{\eps}$ copies of the volume biased ESP in parallel. Using the previous lemma with a  constant probability at least one of the copies finds a non-expanding set.
Moreover, we may bound the time complexity of the algorithm using
\autoref{lem:espsimulate}.
The details of the algorithm is described in Algorithm \autoref{alg:ssefast}.
\begin{algorithm}[htb]
\caption{$\Evo(v,\gamma,\phi,\eps)$ }
\label{alg:ssefast}
\begin{algorithmic}[1]
\STATE Let $S_0\leftarrow \{v\}$,  and $T\leftarrow \eps\ln{\gamma}/6\phi$, and $c$ as defined in \autoref{lem:espsuccessprob}.
\STATE Run $\gamma^{\eps/2}$ independent copies of the volume biased ESP,
using the simulator $\gensam$, starting from $S_0$, in parallel. Stop each copy as soon as the length of its sample path reaches $T$.
\STATE If any of the copies finds a set $S$, of volume $\mu(S)\leq 2\gamma^{1+\eps/2}/c$, and conductance $\phi(S)\leq \sqrt{200(1-\ln c)\phi/\eps})$, stop the algorithm and return $S$.
\end{algorithmic}
\end{algorithm}

Now we are ready to prove \autoref{thm:ssefast}.
\ssefast*
\begin{proof}
Let $U'=U^T$ as defined in \autoref{lem:espsuccessprob}.
First of all, for any $v\in U'$, by \autoref{lem:espsuccessprob}, each copy of volume biased ESP, with  probability $\Omega(\gamma^{-\eps/2})$,  finds  a set $S$ such that  $\mu(S)\leq 2\gamma^{\eps/2}/c$,  and $\phi(S)\leq \sqrt{200(1-\ln c)\phi/\eps}$; but, since $\gamma^{\eps/2}$ copies are executed independently,  at least one of them will succeed with a constant probability.
Therefore, with a constant probability the output set will satisfy properties (1), (2)
in theorems' statement. This proves the correctness of the algorithm. 

It remains to compute the time complexity.
Let, $k:=\gamma^{\eps/2}$ be the number of copies, and  $W_1,\ldots,W_k$ be random variables indicating the work done by each of the copies in a single run of \Evo, thus $\sum_i W_i$ is the time complexity of the algorithm.
Let $M$ be the random variable indicating the volume of the output set of the algorithm; we let $M=0$ if the algorithm does not return any set.
Also, for $1\leq i\leq k$, let $X_i$ be $1/M$ if the output set is chosen from the $i^{th}$ copy, and $0$ otherwise, and let $X:=\sum X_i$.
We write $\kPP{v}{k}{.}$ to denote the probability measure of the
$k$ independent volume biased ESP all started from $S_0=\{v\}$, and $\kEE{v}{k}{.}$ for the expectation.
To prove the theorem it is sufficient to show
$$\kEE{v}{k}{X\sum_{i=1}^k W_i}=O(\gamma^{\eps}\phi^{-1/2}\log^2{n}).$$
 By linearity of expectation, it is sufficient to show that for all $1\leq i\leq k$,
$$\kEE{v}{k}{X_i\sum_{j=1}^k W_j}=O(\gamma^{\eps/2}\phi^{-1/2}\log^2{n}),$$
 By symmetry of the copies, it is sufficient to show that only for $i=1$. Furthermore, since conditioned on $X_1\neq 0$, $W_1 = \max_i W_i$, we just need to show,
$$\kEE{v}{k}{X_1W_1}=O(\phi^{-1/2}\log^2{n}),$$
Let $\tau$ be a stopping time, bounded from above by $T$, indicating the first time where a set $S_\tau$ of volume
$\mu(S_\tau)\leq 2\gamma^{1+\eps}/c$, and conductance $\phi(S_\tau)\leq \sqrt{200(1-\log c)\phi/\eps}$ is observed in the first copy if it is executed up to time $\tau$, and $W_1(\tau)$ as the amount of work done by that time. Observe that for any element of the joint probability space, $X_1W_1 \leq W_1(\tau)/\mu(S_\tau)$. This is because, we always have $W_1\leq W_1{\tau}$, and $X_1\leq \mu(S_\tau)$. Therefore,
$$\kEE{v}{k}{X_1W_1} \leq \kEE{v}{k}{\frac{W_1(\tau)}{\mu(S_\tau)}} =
\hEE{v}{\frac{W(\tau)}{\mu(S_\tau)}}= O(T^{1/2}\log^{3/2}{n}) = O(\phi^{-1/2}\log^2{n}),$$
where the second to last equation follows from  \autoref{lem:espsimulate}.
\end{proof}

\section{Lower Bounds on Uniform Mixing Time of Random Walks}
\label{sec:mixingtime}
In this section we prove lower bounds on the mixing time of reversible markov chains. Since any reversible finite state markov chain can be realized as a random walk on a weighted undirected graph, for simplicity of notations, we model the markov chain as a random walk on  a weighted graph $G$.

 The $\eps$-{\em mixing time} of a random walk in total variation distance is defined as
$$ \tau_{V} (\eps) := \min \left\{t: \sum_{v\in V} |\cPP{u}{X_t = v}-\pi(v)| \leq \eps, \forall u\in V \right\}.$$
The mixing time of the chain is usually defined  as $\tau_{V} (1/4)$.
The $\eps$-{\em uniform} mixing time of the chain is defined as
\begin{eqnarray}
\label{eq:uniformmixing}
 \tau(\eps) := \min\left\{ t:~\left| 1-\frac{\cPP{u}{X_t=v}}{\pi(v)}\right| \leq \eps, \forall u,v\in V\right\}.
 \end{eqnarray}
It is worth noting that the uniform mixing time  can be considerably larger than the  mixing time in total variation distance.

Let $\phi(G) := \min_{S:\mu(S)\leq \mu(V)/2} \phi(S)$. Jerrum and Sinclair \cite{JS89} prove that the $\eps$-uniform mixing time of any {\em lazy} random walk is bounded from above by,
$$ \tau(\eps) \leq \frac{2}{{\phi(G)}^2}\left(\log \frac{1}{\min_{v} \pi(v)} + \log \frac{1}{\eps}\right).$$
On the other hand, one can use $\phi(G)$ as the {\em bottleneck ratio} to provide lower-bound on the mixing time of the random walks. It follows from the Cheeger's inequality that (see e.g. \cite{LPW06}),
$$ \tau_{V}(1/4) \geq \frac{1}{4\phi(G)}. $$

In the next proposition we  prove stronger lower bounds on the uniform mixing time of any reversible markov chain.

\begin{proposition}
\label{prop:bottleneckratio}
For any  graph $G=(V,E)$, $1\geq \gamma\leq \mu(V)/2$, and $0<\eps<1$,
 $$ \tau(\eps) \geq \frac{\ln(\mu(V)/2\gamma)}{\phi(\gamma)}-2.$$
\end{proposition}
\begin{proof}
Let $\U\subseteq V$ such that $\mu(\U)\leq\gamma$, and $\phi(\U)=\phi(\gamma)$.
Let $t\geq -\ln(2\pi(\U))/2\phi(\U)-2$ be an even integer. By the next claim, and equation \eqref{eq:remprob}, there exists a vertex  $u\in \U$ such that
$$   \rem(u,t,\U) \geq \left(1-\phi(\U)\right)^t \geq 2\pi(\U).$$
Since $\cPP{u}{X_t\in \U} \geq \rem(u,t,\U)$, there is a vertex $v\in \U$ such that,
$$ \frac{\cPP{u}{X_t=v}}{\pi(v)} \geq \frac{\cPP{u}{X_t\in \U}}{\pi(\U)}  \geq \frac{2\pi(\U)}{\pi(\U)} = 2, $$
where the first inequality uses $\cPP{u}{X_t\in \U} = \sum_{v\in \U} \cPP{u}{X_t=v}$. Therefore,
$ \frac{|\cPP{u}{X_t=v} - \pi(v)|}{\pi(v)} \geq 1, $
and by equation \eqref{eq:uniformmixing}, for any $\eps<1$,
 $\tau(\eps)  \geq t.$ The proposition follows from the choice of $S$; that is $ t > \frac{\ln(\mu(V)/2\gamma)}{\phi(\gamma)}-2$.
\begin{claim}
For any (weighted) graph $G$, $\U\subseteq V$, and integer $t>0$,
$$ \tp{\bpsi_{\U}} (I_{\U}D^{-1}AI_{\U})^{2t} \b1_{\U} \geq (1-\phi(\U))^{2t}$$
\end{claim}
The proof  is very similar to \autoref{prop:escprob}, except, here $I_SD^{-1}AI_S$ is not (necessarily) a positive semidefinite matrix. This is the reason that we prove the inequality only for {\em even} time steps of the walk.
\begin{proof}
Let $P:=D^{-1/2}I_S A I_S D^{-1/2}$ be a symmetric matrix. By equation \eqref{eq:rightPsqrt}, $1-\phi(S) = \tp{\bpsi_\U} (I_SD^{-1}AI_\U) \b1_\U$.
Using equation \eqref{eq:leftPsqrt}, the claim's statement is equivalent to the following equation:
\begin{equation}
 \tp{\sqrt{\bpsi_\U}} P^{2t} \sqrt{\bpsi_\U} \geq \left(\tp{\sqrt{\bpsi_\U}} P \sqrt{\bpsi_\U}\right)^{2t}
 \label{eq:mixingequi}
 \end{equation}
 The above inequality can be proved using techniques similar to \autoref{lem:matrixnormmon}.
 Let $\bv_1,\ldots,\bv_n$ be the eigenvectors of $P$, corresponding to the eigenvalues $\lambda_1,\ldots,\lambda_n$.
By equation \eqref{eq:xpxexpanding}, \eqref{eq:mixingequi} is equivalent to the following equation:
$$ \sum_{i=1}^n \langle \sqrt{\bpsi_\U},\bv_i\rangle^2 \lambda_i^{2t} \geq \left(\sum_{i=1}^n \langle \sqrt{\bpsi_\U},\bv_i\rangle^2 \lambda_i\right)^{2t}.$$
Since $\sqrt{\bpsi_\U}$ is a norm one vector, and $f(\lambda)=\lambda^{2t}$ is a convex function, the above equation holds by the Jensen's inequality.
\end{proof}
\end{proof}

We remark that the above bound only holds for the uniform mixing time, and it can provide much stronger lower bound than the bottleneck ratio, if  $\gamma \ll \mu(V)$.


\subsection*{Acknowledgements.}
We would like to thank Or Meir and Amin Saberi for stimulating discussions.
We also thank anonymous reviewers for helpful comments on the earlier version of this document.
\bibliographystyle{alpha}
\bibliography{references,hs}
\end{document}